\documentclass[american,aps,prx,reprint,superscriptaddress, longbibliography,floatfix]{revtex4-2}
\usepackage{amsmath}
\usepackage{latexsym}
\usepackage{amssymb}
\usepackage{graphicx}
\usepackage[colorlinks=true, citecolor=blue, urlcolor=blue]{hyperref}
\usepackage{float}
\usepackage{amsfonts}
\usepackage{textcomp}
\usepackage{mathpazo}
\usepackage{comment}
\usepackage{xr}
\usepackage{xfrac}
\usepackage{bbm}
\usepackage{xcolor}
\definecolor{myurlcolor}{rgb}{0,.5,.5}
\definecolor{mycitecolor}{rgb}{0,.6,0}
\definecolor{myrefcolor}{rgb}{2,0,0}
\usepackage{hyperref}
\hypersetup{colorlinks,
linkcolor=myrefcolor,
citecolor=mycitecolor,
urlcolor=myurlcolor}

\usepackage[draft]{fixme}
\usepackage{amsmath,bbm}
\usepackage{graphicx}
\usepackage{amsfonts}
\usepackage{amssymb}
\usepackage{amsmath, amssymb, amsthm,verbatim,graphicx,bbm}
\usepackage{mathrsfs}
\usepackage{color,xcolor,longtable}

\usepackage{xr}
\makeatletter
\newcommand*{\addFileDependency}[1]{
  \typeout{(#1)}
  \@addtofilelist{#1}
  \IfFileExists{#1}{}{\typeout{No file #1.}}
}
\makeatother

\newcommand*{\myexternaldocument}[1]{
    \externaldocument{#1}
    \addFileDependency{#1.tex}
    \addFileDependency{#1.aux}
}

\myexternaldocument{manuscript_PRL}

\newcommand{\beq}[0]{\begin{equation}}
\newcommand{\eeq}[0]{\end{equation}}

\newcommand{\one}{\leavevmode\hbox{\small1\normalsize\kern-.33em1}}

\def\be{\begin{equation}}
\def\ee{\end{equation}}
\def\ben{\begin{eqnarray}}
\def\een{\end{eqnarray}}
\def\eea{\end{array}}
\def\bea{\begin{array}}

\newcommand{\Tr}[1]{\mathrm{Tr}#1}
\newcommand{\bei}{\begin{itemize}}
\newcommand{\eei}{\end{itemize}}
\newcommand{\ket}[1]{|#1\rangle}
\newcommand{\bra}[1]{\langle#1|}

\newcommand{\proj}[1]{\ket{#1}\!\bra{#1}}
\newcommand{\braket}[2]{\langle{#1}|{#2}\rangle}

\renewcommand{\emph}[1]{\textbf{#1}}

\makeatletter
\newtheorem*{rep@theorem}{\rep@title}
\newcommand{\newreptheorem}[2]{%
\newenvironment{rep#1}[1]{%
 \def\rep@title{#2 \ref{##1}}%
 \begin{rep@theorem}}%
 {\end{rep@theorem}}}
\makeatother

\theoremstyle{plain}
\newtheorem{thm}{Theorem}
\newtheorem*{thm*}{Theorem}
\newreptheorem{thm}{Theorem}
\newtheorem{lem}{Lemma}

\newtheorem{prop}[thm]{Proposition}
\newtheorem{cor}[thm]{Corollary}

\newtheorem{defn}[thm]{Definition}

\theoremstyle{definition}

\theoremstyle{remark}

\newcommand{\blk}{\color{black}}

\usepackage[T1]{fontenc}

\begin{document}

\title{Polytopic Quantum Resource Theories: Geometry and Structures}
\author{Moein Naseri}
\email{moein.naseri@ug.edu.pl}
\affiliation{International Centre for Theory of Quantum Technologies, Uniwersytet Gdański, 4 prof. Marii Janion Street,
80-309, Gda\'nsk, Poland}
\author{Chirag Srivastava}
\affiliation{Institute of Informatics, Faculty of Mathematics, Physics and Informatics,
University of Gdansk, Wita Stwosza 57, 80-308 Gdansk, Poland}
\author{Chandan Datta}
\affiliation{Department of Physical Sciences, Indian Institute of Science Education
and Research Kolkata,\\ Mohanpur, West Bengal 741246, India}
\author{John H. Selby}
\affiliation{International Centre for Theory of Quantum Technologies, Uniwersytet Gdański, 4 prof. Marii Janion Street,
80-309, Gda\'nsk, Poland}
\author{Shubhayan Sarkar}
\affiliation{Institute of Informatics, Faculty of Mathematics, Physics and Informatics,
University of Gdansk, Wita Stwosza 57, 80-308 Gdansk, Poland}

\begin{abstract}	
Quantum resource theories provide a unifying framework to quantify, compare, and manipulate quantum resources under well-defined operational constraints. Here, we consider any resource theory where the set of free states can be expressed as a convex combination of a set of quantum states, referred to as extremal states and name them as polytopic quantum resource theories (PQRT). These include some of the most studied resource theories, such as coherence and magic. We formulate a novel tensorial representation of PQRTs that reveals the underlying geometry of these theories and provides insight into the origin of the resources. We further address a fundamental question in resource theories that when two theories should be regarded as physically equivalent, and to this purpose we introduce notions of homomorphism and isomorphism that compare both the structure of free states and the allowed transformations. Using the tools we develop, we find results revealing the geometrical and structural foundations of such theories. Interestingly, we find that all polytopic resource theories with a fixed number of pure extremal points are equivalent under a physical map, up to normalisation. Additionally, we introduce linearly independent polytopic resource theories (resource theory of ``basis-non-convexity''), where the set of extremal free states forms a basis of the quantum density operators. We further study the categorical structures of PQRTs beyond single systems.
\end{abstract}


\maketitle

\section{Introduction}

Identifying and characterising the quantum states and operations that enable advantages in information-processing tasks stands as a central challenge in quantum information science.
A well-known example is entanglement \cite{RMP_Entanglement}, which serves as a key resource for protocols such as quantum teleportation \cite{Teleportation_Bennett}, superdense coding \cite{Dense_Coding}, and quantum key distribution \cite{QKD_Ekert}. However, entanglement is not the only useful quantum resource. In recent years, several other quantum features—such as coherence \cite{Coherence_Baumgratz, RMP_Coherence,selby2020compositional}, imaginarity \cite{Hickey_2018, Resource_Imaginarity, VarunKondra_2023}, purity \cite{Purity_Horodecki,Purity_Gour,Purity_Streltsov}, asymmetry \cite{Asymmetry_Gour}, magic \cite{Magic_Veitch}, and others \cite{RMP_Resource}—have been identified as valuable resources for different quantum information processing tasks. 

Over the years, a systematic framework known as quantum resource theory has been developed to study such quantum resources in a mathematically rigorous way \cite{Resource_Coecke,Resource_Fritz,RMP_Resource}. This framework provides a unified approach for understanding different kinds of quantum resources and how they can be used optimally for various information-processing tasks. The central idea is to characterize quantum systems under a restricted set of operations, called free operations. Since these operations form only a subset of all physically allowed quantum operations, they can generate only certain states, known as free states. Any state that cannot be prepared using free operations is considered a resource state. In this way, a resource theory clearly distinguishes between free and resource states based on operational limitations. For example, in the resource theory of entanglement, if the allowed operations are local operations and classical communication (LOCC) \cite{NLWE_Bennett,Nielsen_PRL}, then entangled states are precisely those that cannot be created using LOCC alone and are therefore regarded as valuable resources \cite{RMP_Entanglement}.

In recent years, many quantum resource theories have been developed and studied extensively, each focusing on a specific type of resource \cite{Coherence_Baumgratz, RMP_Coherence, Hickey_2018, Resource_Imaginarity,VarunKondra_2023,Purity_Horodecki,Purity_Gour,Purity_Streltsov,Asymmetry_Gour,Magic_Veitch,RMP_Resource}. However, relatively few works have aimed at constructing a general framework that can encompass multiple resources within a single unified approach \cite{Horodecki_GQRT,Brandao_GQRT,Korzekwa_GQRT,Liu_GQRT,Takagi_GQRT,Takagi1_GQRT,Fang_GQRT,Regula_GQRT,Vijayan_GQRT,Kuroiwa_GQRT,Sarkar_GQRT}, and none (that we are aware of) allow for a direct comparison of resource theories. Such a general resource-theoretic framework would allow us to study different resources under one umbrella, rather than treating each one separately. Motivated by this idea, we introduce a general framework that we call polytopic quantum resource theories (PQRT). In this approach, the set of free states is defined as the convex combination of a fixed set of quantum states, and the free operations are those that preserve this set of free states. Interestingly, this framework naturally includes several well-known resource theories, such as coherence and magic, as special cases. 

We first present a number of fundamental results for single quantum systems capturing the structural and geometrical foundations of such resource theories, and then extend the framework to composite systems. Importantly, providing an abstract framework in which we can study different resource theories, also provides a way to directly compare resource theories. We introduce a novel tensorial framework for PQRTs that exposes their underlying structure. We further provide notions of resource theory homomorphisms and isomorphisms, both in the single system and composite system cases. To give a simple example, this allows us to show that the choice of basis in defining the resource theory of coherence leads to isomorphic resource theories. 

Similar approaches have been considered in particular to generalise the resource theory of coherence \cite{plenio,Das_2020,Srivastava_2021,banerjee2021quantumcoherenceincompleteset,selby2020compositional}, where free states are defined as a convex combination of various different quantum states rather than the standard basis vectors. Interestingly, we find that all PQRTs with a fixed number of pure extremal points are CP-isomorphic. This means that for two such resources, if the number of extremal free states is the same, then finding relevant quantities in one resource theory is enough to find the same quantities in the other. Thus, all such PQRTs are equivalent under a physical map (up to normalisation). Moreover, this implies that results in some of the previous generalisations of coherence \cite{plenio, Das_2020}, follow straightaway from the resource theory of coherence. As an additional example, we introduce linearly independent PQRTs.

 Now we turn to the basic definition which underpins our work:\blk
\begin{defn}[Polytopic Quantum Resource Theory, PQRT]
   Let $\mathcal{H}$ be a $d$-dimensional Hilbert space and $S_{ver}^{N}:=\{\nu_{i}: \nu_i\in D(\mathcal{H}), i\in\{1,...,N\} \}$ where each element in $S_{ver}^{N}$ cannot be written as a convex combination of the others. We call the elements $\nu_i$ to be vertexal states. We define the PQRT with respect to the set of vertices $S_{ver}^{N}$ as follows.  
\begin{itemize}
    \item \textbf{Set of Free States:} The set of free states is the convex combination of the vertexal states in $S_{ver}^{N}$ i.e.
    \begin{align}
        F_{s}= \text{conv}(S_{ver}^{N}).
    \end{align} 
    
    \item \textbf{ Set of Free Operations:} We consider the set of free operations $F_{O}$ to be maximal in the sense that any completely positive trace preserving (CPTP) operator mapping a free state to another is a free operation i.e.
    \begin{align}
        F_O = \{\Lambda \in \text{CPTP}: \Lambda(\sigma)\in F_{s}, \forall \sigma \in F_{s} \}.
    \end{align}
\end{itemize}
\end{defn}
It is clear from the definition that if $\Lambda \in F_O$ and $\sigma \in F_s$ then:
\begin{align}
    \Lambda(\sigma)=\sum_{j=1}^{N} p_{j}^{\sigma,\Lambda}\nu_{j}
\end{align}
where  $\nu_{j}\in S_{ver}^N$ \blk and $\{p_{j}^{\sigma,\Lambda}\}_j$ is a probability distribution depending on the choice of the free map $\Lambda$ and the free state $\sigma$ such that $\sum_jp_{j}^{\sigma,\Lambda}=1$. If we substitute $F_O$ with $F_O^{CP} = \{\Gamma \in \text{CP}: \Gamma(\sigma)\in \alpha F_{s}, \alpha>0, \forall \sigma \in F_{s} \}$ we will have a weaker form of the PQRT in the sense that the free operations are completely positive but they may not preserve the trace. Stochastic quantum maps (completely positive trace non-increasing maps) are the examples of such operations. 

If the cardinality of $S_{ver}$ is $|\mathbb{R}|$, where $\mathbb{R}$ denotes the real line, then this includes resource theories such as entanglement theory, in which the set of vertexal states consists of pure separable states. Indeed, this set forms a continuous manifold embedded in $\mathbb{R}^d$, for some manifold dimension $d$.

The definition of free operations demands that every free state is mapped to a free state, however, thanks to linearity of CPTP maps, and convexity of the set of free states, we only actually need to check the action on vertexal states. That is:
\begin{lem}\label{lem1new}
For a CPTP map $\Lambda$ we have that 
\[\Lambda(\sigma)\in F_s\ \ \forall \sigma \in F_s\ \ \iff\ \ \Lambda(\nu_l)\in F_s\ \ \forall \nu_l\in S_{ver}^N.\]
\end{lem}
\proof The $\Leftarrow$ direction is trivial as $\nu_l \in F_s$. The $\Rightarrow$ direction follows from the fact that $F_s=\text{conv}(S_{ver}^N)$ we can write $\sigma=\sum_l p^\sigma_l \nu_l$ for some probability distribution $p^{\sigma}_l$, and, as $\Lambda$ is linear we then have that $\Lambda(\sigma)=\Lambda(\sum_l p^\sigma_l \nu_l)=\sum_lp_l^\sigma \Lambda(\nu_l)$. Now, by assumption, $\Lambda(\nu_l)\in F_s$ and so we have a convex combination of things in $F_s$ which itself must live in $F_s$.
\endproof

It then immediately follows that one can check whether or not a given CPTP map is free by the following linear program:
\begin{cor}A given CPTP map $\Lambda$ is a free operation if and only if the following is feasible
    \begin{align}
        \text{Find}\ & \{{p_l^i}\}_{i,l=1,...,N}\\ \nonumber
        \text{s.t.}\ & p_l^i \geq 0 \quad \forall i,l,\\ \nonumber
        & \sum_l p_l^i = 1 \quad \forall i, \\ \nonumber
        & \sum_l p_l^i \nu_l = \Lambda(\nu_i) \quad \forall i.
    \end{align}
\end{cor}
We also have:
\begin{cor} \label{Transformation-Check-SDP}
  The transformation $\rho\to \sigma$ is possible in the PQRT $\mathcal{R}$ if and only if the following, which is an SDP, is feasible:
  \begin{align}
        \text{Find}\ & \Lambda \in CPTP\\ \nonumber
        \text{s.t.}\ & \Lambda(\rho)=\sigma,\\ \nonumber
        & \Lambda(\nu_i)\in F_{s} \quad \forall i.
    \end{align}
\end{cor}
We refer to Appendix for the proofs a fidelity-based characterization of these SDPs in the form of an optimization problem.

Often we want to go beyond just considering free operations, but to consider free instruments which represent the nondestructive measurements which we can perform in our resource theory. A measurement with $N$ outcomes is associated with an instrument consisting of $N$ sub-channels, $\boldsymbol{\Gamma}:=\{\Gamma_n\}_{n=1}^N$ (completely positive trace non-increasing maps)  such that $\Lambda_{\boldsymbol{\Gamma}}:=\sum_{n=1}^N\Gamma_n$ is completely positive and trace preserving. When the outcome is unknown, the post-measurement state is given by $\Lambda_{\boldsymbol{\Gamma}}(\rho)$, given the pre-measurement state is $\rho$. On the other hand, when the measurement outcome is known to be the $n^{\text{th}}$ outcome (presuming this can happen with nonzero probability), the post-measurement state is $\frac{\Gamma_n(\rho)}{\Tr(\Gamma_n(\rho))}$. 
\begin{itemize}
    \item \textit{ \textbf{Set of Free Instruments:} We consider the set of free instruments $F_I$ to be maximal in the sense that any set of completely positive trace nonincreasing operators which sum to a trace preserving operator, in which all postmeasurement states of free states are themselves free, i.e.,
    \[F_I:=\left\{\boldsymbol{\Gamma}:\frac{\Gamma_n(\sigma)}{\Tr(\Gamma_n(\sigma))} \in F_s, \forall \sigma \in F_s, n\in N\in\mathbb{N}\right\}.\]}
\end{itemize}
It is clear from the definition that if $\boldsymbol{\Gamma}\in F_I$ then $\Lambda_{\boldsymbol{\Gamma}}\in F_O$, however, it is not the case that if we have $\Lambda \in F_O$ that any decomposition $\Lambda = \sum_{n=1}^N \Lambda_n$ necessarily defines a valid free instrument, i.e., $\{\Lambda_n\}_{n=1}^N$ is not necessarily in $F_I$.

\section{Tensorial Representation of Polytopic Resource Theories and Geometry}

In this section, we develop a geometric and tensorial formulation of PQRTs that separates the linear structure of the state space from its inner product structure. The key idea is to begin with an abstract finite-dimensional vector space and treat the Hilbert space structure as an additional piece of data, induced by a suitable choice of dual basis. This perspective allows us to represent states as elements of the tensor product $V \otimes V^{*}$, where $V$ represents vector space and $V^{*}$ is its dual, and to interpret density operators as convex combinations of rank-one tensors of the form $\ket{v} \otimes f_v$ with $\ket{v}\in V$ and $f_v \in V^{*}$. Within this framework, the geometry of the theory which is encoded in the inner product, emerges from a positive matrix defining the dual pairing. 

This approach provides a natural way to view PQRTs as intrinsically geometric objects: different choices of dual bases correspond to different inner products, and hence to different but equivalent Hilbert space realizations. In particular, it allows us to formalize when two PQRTs should be considered geometrically equivalent, namely when their vertex sets are related by a unitary transformation, or more generally, when they can be mapped into one another via a change of dual structure. This tensorial viewpoint not only clarifies the role of geometry in defining free states but also provides a flexible framework for comparing and classifying resource theories beyond a fixed Hilbert space representation.

Consider a basis $B=\{\ket{i}\}_{i=0}^{d-1}$ such that $B$ is a set of $d$ linearly independent vectors. We define the vector space $V$ on the field of complex numbers $\mathbb{C}$ as follows:
\begin{align}
    V=\text{span}(B).
\end{align}
Note that we don't have a Hilbert space yet as the inner product is not defined. Furthermore $V$ is a $d$-dimensional vector space. Consider the set of the following linear functionals on the vector space $V$:
\begin{align}
    B'=\{f_{i}:V&\xrightarrow{\text{Linear}}\mathbb{C}, f_i(\ket{i})=1,\nonumber\\ &f_{i}(\ket{j})=(f_{j}(\ket{i}))^{*}=K_{ij}\text{  }i\neq j,\text{ }\forall i,j \}.
\end{align}
We have the following lemma.
\begin{lem}
If the matrix $K\equiv[K_{ij}]$ is invertible, then $B'$ is a basis for the dual space $V^{*}$. \label{Dual-Base}
\end{lem}
    \begin{proof}
        We know that the dual space $V^{*}$ is a $d$-dimensional vector space. It is enough to show that the elements of $\{f_{i}\}_{i}$ are linearly independent. Assume for the sake of contradiction that $f_{i}$ are linearly dependent, hence one can write:
        \begin{align}
            \sum_{j=0}\alpha_{j}f_{j}=0
        \end{align}
        for some $\alpha_{j}\in\mathbb{C}$. Applying the sum on the elements of $B$ we obtain $d$ linear equation as:
        \begin{align}
        \sum_{j}\alpha_{j} K_{jk}=0
        \end{align}
        We can express this set of equation via matrix notations by denoting  $\vec{\alpha}^{T}\equiv [\alpha_0,...,\alpha_{d-1}]$:
        \begin{align}
            K^{T}\vec{\alpha}=0
        \end{align}
        By the definition $K$ is an invertible and thus the eqution does not have a nonzero solution for the vector $\vec{\alpha}$. Hence the vectors $f_{j}$ are linearly independent.
    \end{proof}

    Assuming $K>0$, one can define a scalar product on $V$ and make a Hilbert space out of it.
    \begin{lem}
        Let $\ket{v}=\sum_{i} c_{i}\ket{i} \in V$, $\ket{w}=\sum_{i}c'_{i}\ket{i}\in V$ and $K>0$. The map $\braket{\cdot}{\cdot}_1:V\times V \to \mathbb{C}$ defined by:
        \begin{align}
            \braket{v}{w}_1 \equiv (\sum_{i}c^{*}_{i}f_{i})(\sum_{j}c'_{j}\ket{j})
        \end{align}
        defines a scalar product on the vector space $V$. We denote $f_{v}=\sum_{i}c^{*}_{i}f_{i}$
    \end{lem}
    \begin{proof}
        We shall check whether the map $\braket{\cdot}{\cdot}_1$ has all the properties of a scalar product.
        \begin{itemize}
            \item $\braket{v}{v}= \sum_{i,j} c^{*}_{j}c_{i} k_{ij}= \vec{c}^{\dagger}K\vec{c}\geq 0$.
            \item As $f_{k}$ are linear and $k_{ij}=k_{ji}^{*}$, the map is sesquilinear.
            \item If $\ket{v}=0$ obviously $\braket{v}{v}=0$. Also if $\braket{v}{v}=0$, $K>0$ implies that $\ket{v}=0$.
        \end{itemize}
    \end{proof}

Note that by Riesz representation theorem, any scalar product $\braket{v}{w}$ can be represented by $(\sum_{i}c^{*}_{i}f_{i})(\sum_{j}c'_{j}\ket{j})$ for some $B'$ with $K>0$. Therefore any (separable) Hilbert space can be equivalently represented by a triple $(V,B,B')$ \cite{fabian2011banach}. Now consider the two following sets:
\begin{align}
    B'_1=\{g_{i}:V\xrightarrow{\text{Linear}} \mathbb{C}, g_i(\ket{i})=1, g_{i}(\ket{j})=(g_{j}(\ket{i}))^{*} \notag\\=G_{ij}\text{  }i\neq j,\text{ }\forall i,j, \text{ }G>0 \} \notag \\
    B'_2=\{h_{i}:V\xrightarrow{\text{Linear}} \mathbb{C}, h_i(\ket{i})=1, h_{i}(\ket{j})=(h_{j}(\ket{i}))^{*}\notag\\
    =H_{ij}\text{  }i\neq j,\text{ }\forall i,j, \text{ }H>0 \}
\end{align}
According to Lemma \ref{Dual-Base}, these two sets are bases of the dual space $V^{*}$ and they define scalar products on the vector space $V$. Therefore, there exists an invertible linear map $T$ such that:
\begin{align}
    T(g_{i})=h_{i}, \text{ }\forall i
\end{align}
Moreover, we may equivalently represent the space of the density matrices on $(V,B,B')$ as follows:
\begin{align}
D=\text{conv}(\{\ket{v}\otimes f_{v}: \ket{v}\in V, f_{v}
(\ket{v})=1\}) \label{Conv-Tensor}
\end{align}
with the natural definition of the trace by $\Tr(\ket{v}\otimes f_w)=f_{w}(\ket{v})$ on $D$, where $conv(.)$ denotes convex combination. Note that this equivalence also implies that the functionals $f_{v}$ induce a unique geometry on $D$ via the trace norm (i.e. fix the trace norm of each density matrix $\rho\in D$). As an example we may equivalently represent $\frac{1}{2} \ket{+}\bra{+}+\frac{1}{2}\ket{-}\bra{-}$ as $\frac{1}{2}\ket{+}\otimes f_{+}+\frac{1}{2}\ket{-}\otimes f_{-}$. For a PQRT with the pure vertexal states $S_{ver}=\{\ket{\alpha_{i}}\}_{i=1}^{N}$, we can have a similar representation by first inducing the geometry via the dual basis $B'$ on the set $S_{ver}$ and then constructing the set of the free states as
\begin{align}
F_s=conv(\{\ket{\alpha_{i}}\otimes f_{\alpha_{i}}\})
\end{align}
Note that a valid choice of the dual $B'$ must satisfy $f_{\alpha_{i}}(\ket{\alpha_{i}})=1$. Therefore, a PQRT with pure vertexal states can be represented by a tuple $(S_{ver},B')$, knowing the instructions to construct the set of free states by convex combination. As we see different choice of the dual basis $B'$ induces different geometries and in general could lead to representations of different such resource theories. The freedom in choosing the dual basis leads us to a notion of geometric equivalence of two PQRT which is defined as follows.
\begin{defn}[Geometric equivalence]
    Two polytopic resource theories are said to be geometrically equivalent if and only if there exists a unitary transformation that  bijectively maps their vertexal states.
\end{defn} 

Now assume that in the $d$-dimensional Hilbert space $\mathcal{H}$ and for $N\leq d$, we have two PQRTs defined by the two sets:
\begin{align}
    S_{ver,1}^{N}= \{\proj{\psi_{i}}: i\in\{1,...,N\}\} \notag \\
     S_{ver,2}^{N}= \{\proj{\phi_{i}}: i\in\{1,...,N\}\}.
\end{align}
We have the following proposition.

\medskip

\begin{prop}
    For $N\leq d$, consider the resource theories $( S_{ver,1}^{N},B')$ and $(S_{ver,2}^{N},B')$. Then there exists $B''$ such that the PQRT $(S_{ver,2}^{N},B'')$ is geometrically equivalent to the the PQRT $(S_{ver,1}^{N},B')$.
    \label{GeoEquiv-Pure}
\end{prop}

\begin{proof}
We can construct the set $B''$ in a straightforward manner as follows:
\begin{align}
        B''=\{g_{k}\in V^{*}:g_{j}(\ket{\psi_{i}})=\braket{\phi_{j}}{\phi_{i}}_{B'}, \notag \\
        g_{j}(\ket{\psi^{\perp}_{i}})=\braket{\phi^{\perp}_{j}}{\phi^{\perp}_{i}}_{B'}\}.
    \end{align}
    Note that if $N<d-1$, since the part $g_{j}(\ket{\psi^{\perp}_{i}})=\braket{\phi^{\perp}_{j}}{\phi^{\perp}_{i}}_{B'}$ can be defined in many different ways, $B''$ is not unique. Now if we map the vertexal states of the resource theory represented by $(S_{ver,2}^{N},B'')$ to the vertexal states of the resource theory represented by $(S_{ver,1}^{N},B')$ via $\ket{\phi_{i}}\to \ket{\psi_{i}}$, then the scalar product between the states will be preserved (due to the definition of $B''$), therefore this map must be an isometry. Extending it to the whole Hilbert space we obtain a unitary. 
\end{proof}

We conclude this section by showing that any PQRT admits a triple representation with the following sets: one whose elements belong to a set of bounded operators on a vector space, one which is a basis for this vectors space and the other one which is a dual basis and induces the geometry. 
\begin{thm} \label{PQRT-Rep}
    Any PQRT which is defined by the set of vertexal state $S_{ver}=\{\rho_{i}\in D(\mathcal{H})\}_{i\in I}$ with the index set $I$, can be uniquely (up to unitary transformation) represented by a triple $\big(\{V_{i}=\sum_{k}^{dim(V)}W_{k}^{i} \in \mathcal{D}(V):rank(W_{k}^{i})\leq 1\},B, B'\big)$ where $\mathcal{D}(V)$ is the set of bounded operators on $V$ as a vector space, $B$ is a base for $V$ and $B'$ is the dual base which makes $B$ to an orthonormal basis. The set of free state will then be constructed via:
    \begin{align}
        F_{s}=conv\big(\{V_{i}=\sum_{k}^{dim(V)}W_{k}^{i} \in \mathcal{D}(V):rank(W_{k}^{i})\leq 1\}\big)
    \end{align}
    and the set of free operations will be the maximal set of free maps on $F_{s}$.
\end{thm}
We refer to the Appendix for the proof of the theorem.

\section{Isomorphisms and Structural Theorem}

A critical question for understanding resource theories is when two theories should be regarded as physically equivalent. In the context of PQRTs, this amounts to understanding when different choices of vertex states and their associated convex structures describe the same operational content. In this section, we address this question by introducing notions of homomorphism and isomorphism that compare both the state spaces and the allowed transformations.

\begin{defn}[CPTP-Homomorphism of Polytopic Resource Theories]
Let $
\mathcal R_1=(F_{s,1},F_{O,1})$ and $\mathcal R_2=(F_{s,2},F_{O,2})$
be two PQRTs on finite-dimensional Hilbert spaces $\mathcal{H}_{1}$ and $\mathcal{H}_{2}$ respectively, and $\mathcal{B}(\mathcal{H})$ denotes the space of bounded operators on $\mathcal{H}$.
We say that $\mathcal R_1$ is \emph{CPTP-homomorphic} to $\mathcal R_2$  if there exist maps
\[
M:\mathcal B(\mathcal H_1)\to\mathcal B(\mathcal H_2),
\qquad
Z:F_{O,1}\to F_{O,2}
\]
such that the following conditions hold:
\begin{itemize}
    \item[\textnormal{(i)}] \textbf{State Homomorphism:}
    $M:D(\mathcal{H}_1)\to D(\mathcal{H}_2)$ is a linear map such that
    \[
    M(F_{s,1})\subset F_{s,2}.
    \]

    \item[\textnormal{(ii)}] \textbf{Operation Homomorphism:}
    For every $\Lambda\in F_{O,1}$,
    \[
    Z(\Lambda)\in F_{O,2}.
    \]
    We note that $Z(\Lambda)$ is CPTP.

    \item[\textnormal{(iii)}] \textbf{Intertwining Condition:}
    For all $\Lambda\in F_{O,1}$ and all $\rho\in D(\mathcal{H}_1)$,
    \[
    Z(\Lambda)\bigl(M(\rho)\bigr)=M\bigl(\Lambda(\rho)\bigr).
    \]

\end{itemize}

If there moreover exist inverse maps $M^{-1}$ and $Z^{-1}$ satisfying the above conditions with the roles of $\mathcal R_1$ and $\mathcal R_2$ exchanged, we say that the two PQRTs are \textbf{CPTP-isomorphic}.
\end{defn}

As an example of CPTP-isomorphism, we may consider two  different resource theories of coherence for a qubit defined by the bases $\{\ket{0},\ket{1}\}$ and $\{\ket{+},\ket{-}\}$ respectively. These two PQRTs are CPTP-isomorphic with the isomorphism: $Z(\Lambda)=\mathcal{U}\Lambda\mathcal{U}^{-1}$ and $M=\mathcal{U}$ where $\mathcal{U}$ is a unitary operation satisfying $\mathcal{U}\{\ket{0},\ket{1}\}=\{\ket{+},\ket{-}\}$. We note having the Kraus operators $K_{i}$ for $\Lambda$, each Kraus operator of the map $Z(\Lambda
)$ is obtained by $K_i\to UK_iU^{-1}$. Furthermore using an isometry from a qubit space to a subspace of a qutrit space defined by $\mathcal{V}\{\ket{+},\ket{-}\}=\{\ket{1},\ket{2}\}$ one could ascertain that the tuple of the maps $M=\mathcal{V}$ and $Z(\Lambda)=\mathcal{V}\Lambda \mathcal{V}^{-1}$ is a CPTP-homomorphism from the qubit coherence resource theory with the  basis $\{\ket{-},\ket{+}\}$ to the qutrit coherence resource theory with the  basis $\{\ket{0},\ket{1},\ket{2}\}$.

We now relax the trace preservation requirement by allowing general completely positive maps with normalization, leading to the notion of CP-homomorphisms between PQRTs.

\begin{defn}[CP-Homomorphism of Polytopic Resource Theories]
Let $\mathcal R_1=(F_{s,1},F_{O,1})$ and $\mathcal R_2=(F_{s,2},F_{O,2})$
be two quantum resource theories on finite-dimensional Hilbert spaces $\mathcal{H}_{1}$ and $\mathcal{H}_{2}$ respectively.
We say that $\mathcal R_1$ is \emph{CP-homomorphic} to  $\mathcal R_2$  if there exist maps
\[
M:\mathcal B(\mathcal H_1)\to\mathcal B(\mathcal H_2),
\qquad
Z: F_{O,1}\to L(\mathcal{B}(\mathcal{H}_{2}))
\]
such that the following conditions hold:
\begin{itemize}
    \item[\textnormal{(i)}] \textbf{State Homomorphism:}
    $M:D(\mathcal{H}_1)\to D(\mathcal{H}_2)$ is a linear map such that
    \[
    M(F_{s,1}) \subset F_{s,2}.
    \]

    \item[\textnormal{(ii)}] \textbf{Operation Homomorphism:}
    For every $\Lambda\in F_{O,1}$,
    \[\frac{Z(\Lambda(\cdot))}{\Tr\big(Z(\Lambda(\cdot))\big)}\in F_{O,2}.
    \]

    \item[\textnormal{(iii)}] \textbf{Intertwining Condition:}
    For all $\Lambda\in F_{O,1}$ and all $\rho\in D(\mathcal{H})$,
    \[
    Z(\Lambda)\bigl(M(\rho)\bigr)=\alpha M\bigl(\Lambda(\rho)\bigr).
    \]
for some $\alpha > 0$.
\end{itemize}
If, in addition, there exist inverse maps $M^{-1}$ and $Z^{-1}$ with the same properties but with the roles of $\mathcal R_1$ and $\mathcal R_2$ exchanged, we say that the two PQRTs are \textbf{CP-isomorphic}.
\end{defn}

We first consider CPTP-isomorphisms, corresponding to exact, deterministic physical equivalence, and show that such equivalence is highly rigid, reducing essentially to geometric equivalence of the underlying vertexal sets.

\begin{thm}\label{thm:rigidity}
    There is no CPTP isomorphism between any two polytopic resource theories $\mathcal{R}_{1}$ and $\mathcal{R}_{2}$ which are not geometrically equivalent.
    \label{No-Iso}
\end{thm}
We refer to the Appendix for the proof and the details. We now consider CP-isomorphism, which captures equivalence at the level of probabilistic transformations. Within this broader framework, we establish the following universality result.
\begin{thm}
    All the polytopic resource theories with the pure vertices, $|S_{ver}|\leq d$ and the same number of vertexal states are CP isomorphic.
\label{CP-Iso}

\end{thm}
We refer to the Appendix for the proof and the details. We can further show that the CP equivalence preserves state convertibility up to probabilistic free operations, highlighting CP-isomorphism as the natural notion of operational equivalence in the probabilistic setting.
\begin{thm}
Let $\mathcal{R}_1$ and $\mathcal{R}_2$
be CP-isomorphic PQRTs, with CP-isomorphism $(M,Z)$. If a state $\rho \in D(\mathcal{H}_1)$ can be converted to $\sigma \in D(\mathcal{H}_1)$ by a free operation in $\mathcal{R}_1$, then the state $M(\rho) \in D(\mathcal{H}_2)$ can be stochastically converted to $
M(\sigma) \in D(\mathcal{H}_2)$ by a stochastic free operation in $\mathcal{R}_2$.
\label{Stoch-Iso}
\end{thm}
We refer to the Appendix for the proof and the details. 

\medskip

To illustrate the above abstract definitions, we  consider two PQRTs for a qubit defined by the vertex sets $\mathcal R_{01}: \quad S_{ver,01}:=\{\proj{0},\proj{1}\}$ and
$\mathcal R_{+0}: \quad S_{ver,+0:}=\{\proj{+},\proj{0}\}$ (where $\ket{\pm}=\frac{\ket{0}\pm \ket{1}}{\sqrt{2}}$). 
 The free states are given by convex combinations of the vertices. For $\mathcal R_{01}$:
\begin{align}
\rho^{(01)}_q = q\proj{0}+(1-q)\proj{1}
=
\begin{pmatrix}
q & 0\\
0 & 1\text{-}q
\end{pmatrix}.
\end{align}
For $\mathcal R_{+0}$:
\begin{align}
\rho^{(+0)}_p = p\proj{+}+(1-p)\proj{0}
=
\begin{pmatrix}
1\text{-}\sfrac{p}{2} & \sfrac{p}{2}\\
\sfrac{p}{2} & \sfrac{p}{2}
\end{pmatrix}.
\end{align}
We note that $F_s^{(+0)}$ is the line segment connecting $\proj{0}$ and $\proj{+}$ in the Bloch ball. Furthermore by Lemma~\ref{lem1}, a CPTP map is free if and only if it maps vertex states into the free set. Therefore,
\begin{align}
\Lambda \in F_O^{(01)} &\iff \Lambda(\proj{0}),\Lambda(\proj{1}) \in F_s^{(01)}, \\
\Lambda \in F_O^{(+0)} &\iff \Lambda(\proj{0}),\Lambda(\proj{+}) \in F_s^{(+0)}.
\end{align}

The two theories, considering a fixed Hilbert space, are not geometrically equivalent since no unitary can map $\{\ket{0},\ket{1}\}$ to $\{\ket{0},\ket{+}\}$. Hence, by the CPTP rigidity theorem (Thm.~\ref{thm:rigidity}) we have that $\mathcal R_{01} \not\cong_{\mathrm{CPTP}} \mathcal R_{+0}$. However they are CP-isomorphic. We construct the isomorphism explicitly. Define the state map
\begin{align}
M(\rho) = \frac{(I \otimes T)(\rho)}{\mathrm{Tr}[(I \otimes T)(\rho)]},
\end{align}
where $T$ maps dual bases such that (up to unitary transformation) $M(\proj{0}) = \proj{0}$ and  
$M(\proj{1}) = \proj{+}$ ( We refer to the Appendix \ref{CP-Iso-App} for the details of the construction of these maps). Thus, for any free state in $\mathcal R_{01}$, $\rho = p\proj{0}+(1-p)\proj{1}$,
we have
\begin{align}
M(\rho) = p\proj{0}+(1-p)\proj{+} \in F_s^{(+0)}.
\end{align}
The operation map is given by $Z(\Lambda) = M' \Lambda M'^{-1}$ with $M' = I \otimes T$. For instance for the bit-flip channel
\begin{align}
\Lambda(\rho) = p\rho + (1-p) X\rho X,
\end{align}
which is free in $\mathcal R_{01}$ we have
\begin{align}
Z(\Lambda):
\begin{cases}
\proj{0}&\mapsto \ \ p\proj{0} + (1-p)\proj{+},\\
\proj{+} &\mapsto \ \ p\proj{+} + (1-p)\proj{0},
\end{cases}
\end{align}
which is also free in $R_{+0}$.

For the reset channel $\Lambda(\rho) = \proj{0}$
which is free and CPTP in $\mathcal R_{01}$, 
the induced map $Z(\Lambda)$ is CP but not trace-preserving (e.g. $Z(\Lambda)(\ket{+}\bra{+})=2\ket{0}\bra{0}$). After normalization, $\widetilde{\Lambda} = \frac{1}{c_\Lambda} Z(\Lambda)$ where $c_\Lambda:=\max_{\tau\in D(\mathcal H_2)} \Tr\bigl(Z(\Lambda)(\tau)\bigr)$,
one obtains a stochastic free operation satisfying
\begin{align}
\widetilde{\Lambda}(\proj{+}) = \frac{1}{2}\proj{0}.
\end{align}

\section{Golden State and Measures}

In this section, we examine the concept of  golden states in PQRTs and introduce  robustness, geometric, and negativity based measures of resourcefulness. 

\noindent\textbf{Golden states:} A golden state is a universal state from which all other states can be obtained via free operations. We demonstrate that such a state need not exist in general, and that, if it does, it must be pure and satisfy a symmetry condition with respect to the vertexal states. We also define the robustness of resource and study its behavior under free and selective free operations. We begin by the following proposition.

\begin{prop}
The golden state does not necessarily exist in a polytopic resource theory. \label{Non-Existence-Golden}
\end{prop}
We further can show that:
\begin{prop}
    A golden state, if it exists,
    \begin{enumerate}
        \item is pure,
        \item if in addition the vertex states are linearly independent pure states, it has the same fidelity with all the  free states.
    \end{enumerate}
    \label{Golden-Prop}
    \end{prop}
We refer to the Appendix for the proofs and the details. From these results, we immediately conclude that if the PQRT is defined via the pure vertices $\{\proj{\nu_{i}}\}$, then we can represent the golden state as $\ket{\Psi_{G}}=\sum_{i}C_{i}\ket{\nu_{i}}$ with $C_{i}\neq 0$.

\begin{lem}
Let $|\braket{\nu_i}{\nu_j}|=a$ for all $\ket{\nu_i},\ket{\nu_j}\in S_{ver}$, $\ket{\nu_i}\neq \ket{\nu_j}$ and $\ket{\Psi_{G}}=\sum_{i}C_{i}\ket{\nu_i}$. Then all $C_{k}$ are equal.  If further $\langle \nu_i|\nu_j\rangle = a \in \mathbb{R}$ then $|\Psi_{G}\rangle=
\frac{e^{i\varphi}}{\sqrt{d\bigl(1+(d-1)a\bigr)}}\sum_i |\nu_i\rangle$ for some phase $\varphi$. 
\end{lem}
\begin{proof}
For each fixed $k$, we have $\braket{\nu_k}{\Psi_{G}} = \sum_{i=1}^n C_i\braket{\nu_k}{\nu_i}$. Using $\braket{\nu_k}{\nu_k} = 1$ and $\braket{\nu_k}{\nu_i} = a$ for $i \neq k$, this becomes 
\begin{equation*}
\braket{\nu_k}{\Psi_{G}}=\sum_{i \neq k} C_i\, a + C_k = a \sum_{i \neq k} C_i + C_k.
\end{equation*}
Since this is equal to a constant $C$ for every $k$, we obtain $
a \sum_{i \neq k} C_i + C_k = C$. Now write $\sum_{i \neq k} C_i = \sum_{i=1}^n C_i - C_k$. Substituting, we get $a\Big(\sum_{i=1}^n C_i - C_k\Big) + C_k = C$,
which simplifies to $a \sum_{i=1}^n C_i - C = (a - 1) C_k$. Therefore,
\begin{equation*}
C_k = \frac{a \sum_{i=1}^n C_i - C}{a - 1}.
\end{equation*}
The right-hand side is independent of $k$, so all $C_k$ are equal. If  $\langle \nu_i|\nu_j\rangle = a \in \mathbb{R}$ then $|\Psi_{G}\rangle = C\sum_i |\nu_i\rangle$ and the normalization condition gives 
\begin{align}
\langle \Psi_{G} | \Psi_G \rangle=
|C|^2 \sum_{i,k}\langle \nu_i|\nu_k\rangle
=|C|^2\, d\bigl(1+(d-1)a\bigr)
=1.
\end{align}
Hence, $|C|^2 = \frac{1}{d\bigl(1+(d-1)a\big)}
$ and \[|\Psi_G\rangle
=\frac{e^{i\varphi}}{\sqrt{d\bigl(1+(d-1)a\bigr)}}\sum_i |\nu_i\rangle\] for some phase $\varphi$.
\end{proof}

\medskip

\medskip

\noindent\textbf{Robustness measure:}
The robustness of a resource in a PQRT for a quantum state, $\rho \in D(\mathbb{C}^d)$, is defined as  
\begin{equation}\label{albela}
    \mathcal{M}_R(\rho):=\min_{\delta\in D(\mathbb{C}^d)} \left\{r\geq0\Big|\frac{\rho+r\delta}{1+r}\in F_s\right\}. 
\end{equation}
This definition is motivated by the definition of generalised robustness of coherence~\cite{Napoli16,Piani16,Mukhopadhyay19} and entanglement~\cite{Vidal99,Steiner03}. 
Notice that for $\mathcal{M}_R(\rho)$ to be a meaningful quantity it is necessary to consider the minimization over all states in $D(\mathbb{C}^d)$  and not over the free states, since for the latter case there exist PQRTs such that  $\mathcal{M}_R(\rho)$ diverges for any $\rho \not\in F_s$. For example, consider a PQRT where $\mathcal{S}_{\mathrm{ext}}=\{\proj{0},\proj{1}\}$, the set of free states is given by  $p\proj{0}+(1-p)\proj{1},$ $0\leq p\leq 1.$ An arbitrary qubit, $\rho$, can be expressed as $q\proj{0}+(1-q)\proj{1}+c|0\rangle\!\langle 1|+c^*|1\rangle \!\langle 0|$ with the constraints $0\leq q\leq 1$ and $q(1-q)-|c|^2\geq 0$. Now, from simple algebra, it is easy to see that any mixture of a state from the free set with $\rho$ will never yield another free state until the probability of $\rho$ in the mixture is zero. This implies that the robustness of the resource for $\rho$ will diverge when the minimization is taken over the free set of density matrices. 

From the  definition, it is clear that $\mathcal{M}_R(\rho)=0$ iff $\rho\in F_s$. Next, we prove that: 
\begin{prop}
    Robustness of resource is a monotone under the set of free operations.
\end{prop}
\proof
Let $\Lambda_F$ denote a free operation. We want to prove that $\mathcal{M}_R(\Lambda_F(\rho))\leq \mathcal{M}_R(\rho)$. From Eq. \ref{albela}, there exist $\sigma\in F_s$ and $\delta\in D(\mathbb{C}^d)$ for a state $\rho\in D(\mathbb{C}^d)$ such that
\begin{equation}
    \rho=(1+\mathcal{M}_R(\rho))\sigma-\mathcal{M}_R(\rho)\delta.
\end{equation}
This implies
\begin{eqnarray}\label{pawn}
   &&\Lambda_F(\rho) = (1+\mathcal{M}_R(\rho))\Lambda_F(\sigma)-\mathcal{M}_R(\rho)\Lambda_F(\delta) \nonumber \\
   \implies&& \Lambda_F(\sigma)=\frac{\Lambda_F(\rho)+\mathcal{M}_R(\rho)\Lambda_F(\delta)}{1+\mathcal{M}_R(\rho)}.
\end{eqnarray}
From the definition of free operations, $\Lambda_F(\sigma)\in F_s$, thus $\frac{\Lambda_F(\rho)+\mathcal{M}_R(\rho)\Lambda_F(\delta)}{1+\mathcal{M}_R(\rho)}\in F_s$.
Therefore, the quantity $\mathcal{M}_R(\Lambda_F(\rho)):=\min_{\delta\in D(\mathbb{C}^d)}\left\{r\geq0\Big|\frac{\Lambda_F(\rho)+r\delta}{1+r}\in F_s\right\}\leq \mathcal{M}_{R}(\rho)$, since from Eq. \eqref{pawn} it is shown that there exists a state $\Lambda_F(\delta)$ for $r=\mathcal{M}_R(\rho)$ such that $\frac{\Lambda_F(\rho)+r\Lambda_F(\delta)}{1+r}\in F_S$.
 \endproof

Till now, we have considered the monotonicity of robustness of the resource with respect to any free CPTP operations. In the following, we show that the robustness of the resource measure is monotonic, on average, under free instruments. 

\begin{prop} \label{Instrument-Monotone} Robustness is an ``on average'' monotone for free instruments, that is, for all $\boldsymbol{\Gamma}=\{\Gamma_n\}_{n=1}^N \in F_I$ we have that \[\sum_{n=1}^N\Tr(\Gamma_n(\rho))\mathcal{M}_R\left(\frac{\Gamma_n(\rho)}{\Tr(\Gamma_n(\rho))}\right)\leq \mathcal{M}_R(\rho)\] for any state $\rho$. \end{prop}

We refer to the Appendix for the proof. Let us now study the class of PQRTs in which the vertexal states linearly independent and pure.

\section{Resource theory of basis ``non-convexity''}
Consider the linearly independent set of quantum density matrices $S_{ver}=\{\proj{\nu_i}\}$ such that $\proj{\nu_i}\in D(\mathbb{C}^d)$. As they form a basis for $F_s$, the resource in this theory are states that are non-convex with respect to this basis. Thus, we refer to it as the resource theory of basis ``non-convexity''. The maximum cardinality of this set is $d^2$. 

Next, let us try to provide a geometric measure of resourcefulness for this class of QPRTs:

\noindent{\bf Geometric measure:} For this purpose, we define the measure as the minimum distance of the quantum state to $F_s$. Mathematically, it is given by
\begin{equation}\label{measure1}
    \mathcal{M}_F(\rho)=1-\max_{\sigma\in F_s}\mathcal{F}(\rho,\sigma)
\end{equation}
where $\mathcal{F}(\rho,\sigma)$ is the fidelity between quantum states $\rho$ and $\sigma$.

It is quite clear that for free states, the measure is $0$. We now show that it is a resource monotone:
\begin{prop} The geometric measure is a monotone under the set of free operations.
\end{prop}
\begin{proof}
   Let us begin with the measure $\mathcal{M}_F$. It is straightforward to prove that it is non-increasing under free operations as
   \begin{align}
        \max_{\sigma\in F_{s}}\mathcal{F}(\Lambda(\rho),\sigma)&\geq \mathcal{F}(\Lambda(\rho),\Lambda(\sigma^{*}))\notag \\ &\geq \mathcal{F}(\rho,\sigma^{*})=\max_{\sigma\in F_s}\mathcal{F}(\rho,\sigma)
    \end{align}
    Note that $\Lambda(\sigma^{*})$ is a free state as $\Lambda$ is a free map.
\end{proof}

For the particular case, when the cardinality of $S_{ver}$ is $d^2$ we can define
another measure.

\noindent {\bf Negativity measure:} For this purpose, let us note that $\mathcal{S}_{\mathrm{ver}}$ forms a basis of $D(\mathbb{C}^d)$ and thus any quantum state $\rho\in D(\mathbb{C}^d)$  can be uniquely expressed as linear combination of $\proj{\nu_i}$, that is, $\rho=\sum_i\alpha_i(\rho)\proj{\nu_i}$ such that $\sum_i\alpha_i(\rho)=1$ and $\alpha_i(\rho)$ are real \footnote{$\rho=\rho^{\dagger}\implies\sum_i\alpha_i(\rho)\proj{\nu_i}=\sum_i\alpha_i(\rho)^*\proj{\nu_i}\implies\sum_i(\alpha_i(\rho)-\alpha_i(\rho)^*)\proj{\nu_i}=0\implies\alpha_i(\rho)=\alpha_i(\rho)^*$ as $\proj{\nu_i}$ are linearly independent.}. Note that the $\alpha_i(\rho)$ can be negative, denoting, for a given state, the negative coefficients as $\alpha^-_i(\rho)$ we define the resource negativity as
\begin{equation}\label{measure2}
    \mathcal{M}_N(\rho)=-\sum_i\alpha^-_i(\rho).
\end{equation}
It is trivial to check that it is $0$ for free states. Let us now show the monotonicity of the defined measure under the free operations.

\begin{prop}The resource negativity is a  monotone under the set of free operations. \label{Negativity-Monotone}
\end{prop}
We refer to the Appendix for the proof of the proposition.
An immediate observation from this proposition is that we can write $\mathcal{M}(\rho)= \min_{\sigma\in F_{s}} \sum_{i}|\alpha_{i}^{-}(\rho)-\alpha_{i}^{-}(\sigma)|$ and further that $D_{N}(\rho, \sigma):=\sum_{i}|\alpha_{i}^{-}(\rho)-\alpha_{i}^{-}(\sigma)|$ satisfies the properties of a quasi-distance.

\section{Beyond single systems}

In the previous sections we considered free states of a single quantum system, $\mathcal{H}$, and free transformations acting on that system. In this section we extend the framework of PQRTs to composite systems, and transformations between different systems. This follows the ``categorical'' approach to resources developed in Ref.~\cite{Resource_Coecke} where the essential idea is that if we have two different free resources $\sigma_1$ and $\sigma_2$ then the composite $\sigma_1\otimes\sigma_2$ should also be free, similarly, if we have two different free operations $F_1$ and $F_2$ then applying these in sequence $F_2\circ F_1$ and in parallel $F_1\otimes F_2$ should both be free operations. Abstractly this can be captured by the idea that we have an enveloping theory consisting of all quantum operations, with a specified subtheory of free operations -- mathematically, the enveloping theory can be taken as the symmetric monoidal category of CPTP maps, and with a specified symmetric monoidal subcategory of free transformations. 

\begin{defn}[Compositional Polytopic Quantum Resource Theory] 
For each quantum system $\mathcal{H}$, where $\mathcal{H}$ is a finite-dimensional Hilbert space, we choose $S_{ver}^\mathcal{H}:=\{\nu_i : \nu_i \in D(\mathcal{H}), i\in\{1,...,N_\mathcal{H}\})\}$. Such that if $\mathcal{H}=\mathcal{K}\otimes \mathcal{L}$ then $S_{ver}^\mathcal{K}\otimes S_{ver}^\mathcal{L} \subset \text{conv}(S_{ver}^\mathcal{H})$. We then define the resource theory as:
\begin{itemize}
    \item {\bf Set of Free States:} The set of free states for a system $\mathcal{H}$ is
    \beq F_s^\mathcal{H}=\text{conv}(S_{ver}^\mathcal{H})
    \eeq
    \item {\bf Set of Free Operations from $\mathcal{H}$ to $\mathcal{K}$:} We consider the set of free operations $F_O^{\mathcal{H}\to\mathcal{K}}$ to be maximal, i.e., 
    \begin{align}
F&_O^{\mathcal{H}\to\mathcal{K}} := \nonumber \\
&\left\{ \Lambda \in \textnormal{CPTP}^{\mathcal{H}\to\mathcal{K}} \middle| \Lambda\otimes \mathbb{I}_\mathcal{L} (\sigma) \in F_S^{\mathcal{K}\otimes \mathcal{L}}, \forall \sigma \in F_S^{\mathcal{H}\otimes\mathcal{L}}\right\}.
     \end{align}
\end{itemize}
\end{defn}

It immediately follows from this definition that the free operations form a symmetric monoidal subcategory of the category of CPTP maps. That is, that the identity and swap CPTP maps are free, and that sequential and parallel composition of free transformations are themselves free operations. This is crucial for being able to leverage any of the results in Ref.~\cite{Resource_Coecke}. Note that the assumption, that if $\mathcal{H}=\mathcal{K}\otimes \mathcal{L}$ then $S_{ver}^\mathcal{K}\otimes S_{ver}^\mathcal{L} \subset \text{conv}(S_{ver}^\mathcal{H})$, is crucial for this, as is the assumption in our definition of free operations that they preserve free states in a ``complete'' sense, that is, even when just acting locally on part of a bipartite system.

The fact that free operations are characterised by being completely-preserving of free states, however, is a challenge to our characterisation via a linear program, as, in general at least, we must check that we obtain a valid pure state for every $\nu\in \bigcup_{\mathcal{L}}S_{ver}^{\mathcal{H}\otimes \mathcal{L}}$ which will be an infinite set in all but trivial cases. Nonetheless, there are special cases in which this characterisation does simplify back to a linear program. For example, if we have $S_{ver}^{\mathcal{H}\otimes\mathcal{L}}=S_{ver}^\mathcal{H}\otimes S_{ver}^\mathcal{L}$ then our completely-preserving condition simplifies to a preserving condition, namely,  just that $\Lambda(\nu_i)\in F_S^\mathcal{K}$ for all $\nu_i \in F_S^\mathcal{H}$. This is exactly the same as the previous linear program.

Next we can see how to define a notion of CPTP-Homomorphism at this compositional level.
\begin{defn}(CPTP-Homomorphism of Compositional Polytopic Resource Theories)
Let 
\[
\mathcal{R}_1=(\{F_{s,1}^{\mathcal{H}_1}\},\{F_{O,1}^{\mathcal{H}_1\to \mathcal{K}_1}\}),\quad \mathcal{R}_2=(\{F_{s,1}^{\mathcal{H}_2}\},\{F_{O,1}^{\mathcal{H_2\to K_2}}\})
\]
be two compositional PQRTs.

We say that $\mathcal{R}_1$ is \textbf{CPTP-homomorphic}  to $\mathcal{R}_2$ if there exists maps
\begin{align*} \xi&:\{\mathcal{H}_1\} \to \{\mathcal{H}_2\},\\ M_{\mathcal{H}}&:\mathcal{H}\to\xi(\mathcal{H}),\\  Z&:F_{O,1}^{\mathcal{H\to K}}\to F_{O,2}^{\xi(\mathcal{H})\to\xi(\mathcal{K})}
\end{align*}
such that the following conditions hold:
\begin{itemize}
    \item[\textnormal{(i)}] \textbf{State Homomorphism:} $M_\mathcal{H}:\mathcal{B(H)}\to \mathcal{B}(\xi(\mathcal{H}))$ is a CPTP map such that 
    \[M_\mathcal{H}(F_{s,1}^\mathcal{H})\subseteq F_{s,2}^{\xi(\mathcal{H})}.\]
    \item[\textnormal{(ii)}] \textbf{Operation Homomorphism:} For every $\Lambda \in F_{O,1}^{\mathcal{H}\to\mathcal{K}}$,
    \[Z(\Lambda)\in F_{O,2}^{\xi(\mathcal{H})\to\xi(\mathcal{K})}.\]
    We note that $Z(\Lambda)$ is CPTP.
    \item[\textnormal{(iii)}] \textbf{Intertwining Condition:} For all $\Lambda\in F_{O,1}^{\mathcal{H}\to\mathcal{K}}$ and all $\rho\in D(\mathcal{H})$,
    \[Z(\Lambda)(M_\mathcal{H}(\rho))=M_\mathcal{K}(\Lambda(\rho)).\]
    \item[\textnormal{(iv)}] \textbf{Compositionality Conditions:}
    \[ M_{\mathcal{H}\otimes\mathcal{K}}=M_\mathcal{H}\otimes M_\mathcal{K},\]
    \[Z(\Lambda_2\circ\Lambda_1)=Z(\Lambda_2)\circ Z(\Lambda_1),\]
    \[Z(\Lambda_2\otimes\Lambda_1)=Z(\Lambda_2)\otimes Z(\Lambda_1).\]
\end{itemize}
In the above definition, if furthermore the maps are all invertible and the inverses satisfy these conditions, then we say that $\mathcal{R}_1$ and $\mathcal{R}_2$ are \textbf{CPTP-isomorphic}.
\end{defn}
This definition, on the one hand, can be seen to be a generalisation of the definition in the single system case, and, on the other hand, naturally falls out from the categorical point of view  \footnote{Specifically, if we view these resource theories as being defined by symmetric monoidal subcategories, $\textrm{F}_i$ $i\in\{1,2\}$, of the category of CPTP maps giving us monoidal inclusion functors $\iota_{i}:\textrm{F}_i\to \textrm{CPTP}$, then a homomorphism is by a monoidal functor $Z:\textrm{F}_1\to\mathrm{F}_2$ together with a monoidal natural transformation $M:\iota_1\Rightarrow (\iota_2\circ Z)$. This defines a category $\mathbf{Sub}[\textrm{CPTP}]$ where objects are $\iota:\textrm{F}\to \textrm{CPTP}$ and morphisms are pairs $(Z,M)$, the notion of \textbf{CPTP-isomorphic} resource theories is simply isomorphism of objects within this category. See Ref.~\cite{Operadic_Resource_Theories} for more details.}.

\section{Discussion}
In this work, we unified a large class of resource theories under a single framework, in particular, any resource theory where the set of free states is a convex mixture of a finite number of quantum states. We introduced a novel tensorial framework for representing PQRTs that uncovers their underlying geometric structure and sheds light on the origin of the associated resources. We further addressed a fundamental question in resource theories: under what conditions two theories should be considered physically equivalent. To this end, we introduced the notions of homomorphism and isomorphism, which compare both the structure of free states and the set of allowed transformations. Using the tools we developed in this work, we obtained results that uncover the geometric and structural principles underlying these resource theories. As an unexpected result, we found that for any two resource theories with a fixed number of pure extremal states, there exists a CP-isomorphism between the two theories. We also extended the framework to multipartite systems and studied the categorical structure of PQRTs beyond single systems. Additionally, we also introduce a PQRT called resource theory of ``basis no-convexity'' where the extremal free states are linearly independent and thus form a basis. In all the above resource theories, the set of free states is convex, and thus the corresponding resource provides an advantage in a channel discrimination task \cite{adesso}. Moreover, a large class of them also show an advantage in a prepare-and-measure scenario where no component of the experiment is trusted except the dimension \cite{Sarkar_GQRT}.

Several follow-up problems arise from our work. The most important one will be whether a similar equivalence as above among resource theories with continuous extremal free states can be established or not. In the spirit of the works \cite{naseri2022entanglement,naseri2025quantum,fellous2025scalable}, a promising direction is to apply the PQRT representation introduced in this work to the study of quantum resources in quantum computing. Such representations highlight the geometric structure underlying these theories and provide deeper insight into the origin and nature of quantum resources. Another interesting avenue could be to generalise the notions equivalence beyond quantum resource theories to resource theories in a generalised probabilistic theory. A difficult problem in this direction would be to demonstrate equivalences between some non-convex resource theories in a similar manner as done in this work. Finally, pushing forwards the notion of resource theory homomorphism and isomorphism to broader settings will be important for the study and classification of quantum resources.

\section{Acknowldegments}

 S.S. and C.S. acknowledge the
support of the National Science Centre, Poland, under grant Opus 25 no. UMO-2023/49/B/ST2/02468.
 JHS and MN were funded by the European Commission by the QuantERA project ResourceQ under the grant agreement UMO2023/05/Y/ST2/00143. 
MN would like to thank Jan Peszek (Institute of Applied Mathematics and Mechanics,
University of Warsaw) for providing insights related to the functional analytic aspects of this work. JHS would like to thank Gaurang Agrawal and Matt Wilson for insightful discussions on compositional aspects of this work.

\bibliography{references}

\appendix

\section{Fidelity-Based SDP for Distinguishing Free Operations in PQRTs}

A straightforward way to characterise free operations can be summarised in the following lemma.
\begin{lem}\label{lem1}
    Consider $\Gamma \in \text{CPTP}$ with the Kraus representation $\{L_{i}^{\Gamma}\}$. If $\forall \nu_l \in S_{ver}^{N}$ there exists a set $\{p_{\Gamma,{m}}^{l}: p_{\Gamma,{m}}^{l}\geq 0, \sum_m p_{\Gamma,{m}}^{l}=1, m\in \{1,...,N\}\}$ such that:
    \begin{align}
        \sum_i L_{i}^{\Gamma}\nu_{l}L_{i}^{\Gamma\dagger}=\sum_m p_{\Gamma,{m}}^{l}\rho_{m},
    \end{align}
    then $\Gamma\in F_O$.
\end{lem}
This lemma states that if a CPTP map sends every vertexal state to a free state, then it is necessarily a free map; otherwise, it cannot be considered free.
\begin{proof}
    We must show that the channel $\Gamma$ with the given Kraus representation maps any free state to another free state. Let $\sigma$ be an arbitrary free state and $\sigma=\sum_j P_{j}^{\sigma} \nu_{j}$ where $\{P_{j}^{\sigma}\}$ is a probability distribution depending on the state $\sigma$. We have:
    \begin{align}
        \sum_i L_{i}^{\Gamma}\sigma L_{i}^{\Gamma\dagger}=\sum_j P_{j}^{\sigma}\sum_{i}L_{i}^{\Gamma}\nu_{j}L_{i}^{\Gamma \dagger}\notag \\
        =\sum_{j,m}P_{j}^{\sigma}p_{\Gamma,{m}}^{j}\nu_{m}.
    \end{align}
    It is clear that $\{P_{j}^{\sigma}p_{\Gamma,{m}}^{j}\}_{j,m}$ is a probability distribution, thus $\sum_{j,m}P_{j}^{\sigma}p_{\Gamma,{m}}^{j}\rho_{m} \in F_s$.
\end{proof}
Lemma~\ref{lem1} can also be reformulated as a semidefinite program using the fidelity function in the form of an optimization problem. Thus, we have the following corollary in this regard.
\begin{cor}\label{cor:fidelity-sdp}
Let $\Gamma$ be a CPTP map. Then $\Gamma\in F_O$ if and only if for every $\nu_{i}\in S_{ver}$ with $|S_{ver}|=N$, the optimal value of the following semidefinite program is equal to $1$:
\begin{align}
\max_{\{P_j^{(i)}\},\,X_i} \quad & \frac12 \Tr(X_i+X_i^\dagger) \\
\text{s.t.}\quad &
\begin{pmatrix}
\Gamma(\nu_i) & X_i\\
X_i^\dagger & \sum_{j=1}^N P_j^{(i)}\nu_j
\end{pmatrix}\ge 0, \notag\\
& P_j^{(i)}\ge 0 \qquad \forall j, \notag\\
& \sum_{j=1}^N P_j^{(i)}=1. \notag
\end{align}
\end{cor}

\begin{proof}
We use two facts. First, by Lemma~\ref{lem1}, a CPTP map $\Gamma$ belongs to $F_O$ if and only if it maps every vertex state $\rho_i$ into the free set $F_s$. Second, for any density operators $A$ and $B$, the function $\mathcal{F}(A,B) =\Tr\sqrt{\sqrt{A}B\sqrt{A}}$ admits the SDP representation
\begin{align}
\mathcal{F}(A,B)=\max_X \quad & \frac12 \Tr(X+X^\dagger) \\
\text{s.t.}\quad &
\begin{pmatrix}
A & X\\
X^\dagger & B
\end{pmatrix}\ge 0. \notag
\end{align}
Moreover, $0\le \mathcal{F}(A,B)\le 1$, and $\mathcal{F}(A,B)=1$ if and only if $A=B$ \cite{watrous2018theory}.

Assume first that $\Gamma\in F_O$. Since $\nu_i\in F_s$ for every $i$, it follows that $\Gamma(\nu_i)\in F_s$. Hence, for each $i$, there exists a probability distribution $\{P_j^{(i)}\}_{j=1}^N$ such that
\[
\Gamma(\nu_i)=\sum_{j=1}^N P_j^{(i)}\nu_j.
\]
Therefore
\[
\mathcal{F}\!\left(\Gamma(\nu_i),\sum_{j=1}^N P_j^{(i)}\nu_j\right)=1.
\]
By the SDP representation of fidelity, the corresponding SDP has feasible value $1$. Since fidelity between density operators is always bounded above by $1$, the optimal value of the SDP is exactly $1$.

Conversely, assume that for every $i$, the optimal value of the above SDP is $1$. Then for each $i$ there exists a probability distribution $\{P_j^{(i)}\}_{j=1}^N$ such that
\begin{equation*}
\mathcal{F}\!\left(\Gamma(\nu_i),\sum_{j=1}^N P_j^{(i)}\nu_j\right)=1.
\end{equation*}
Since $\Gamma(\nu_i)$ and $\sum_{j=1}^N P_j^{(i)}\nu_j$ are both density operators, the equality condition for fidelity implies
\begin{equation*}
\Gamma(\nu_i)=\sum_{j=1}^N P_j^{(i)}\nu_j \in F_s.
\end{equation*}
Thus $\Gamma$ maps every vertex state into $F_s$. By Lemma~\ref{lem1}, this implies that $\Gamma\in F_O$.
\end{proof}

\section{Proof of corollary \ref{Transformation-Check-SDP}}
By definition, the transformation $\rho\to\sigma$ is possible in the PQRT $\mathcal R$ if and only if there exists a free operation $\Lambda\in F_O$ such that
\begin{align}
    \Lambda(\rho)=\sigma.
\end{align}
Using the definition of free operations,
\begin{align}
    F_O=\{\Lambda\in \mathrm{CPTP}: \Lambda(\tau)\in F_s \ \forall \tau\in F_s\},
\end{align}
this is equivalent to the existence of a CPTP map $\Lambda$ satisfying
\begin{align}
    \Lambda(\rho)&=\sigma,\\
    \Lambda(\tau)&\in F_s \qquad \forall \tau\in F_s.
\end{align}
By Lemma~\ref{lem1new}, the condition that $\Lambda$ maps every free state into $F_s$ is equivalent to checking it only on the vertexal states:
\begin{align}
    \Lambda(\nu_i)\in F_s \qquad \forall i.
\end{align}
Hence, $\rho\to\sigma$ is possible if and only if there exists a CPTP map $\Lambda$ such that
\begin{align}
    \Lambda(\rho)&=\sigma,\\
    \Lambda(\nu_i)&\in F_s \qquad \forall i.
\end{align}

To see that this is an SDP feasibility problem, let $J_\Lambda$ denote the Choi operator of $\Lambda$ \cite{watrous2018theory}. Then $\Lambda$ is completely positive if and only if
\begin{align}
    J_\Lambda \geq 0,
\end{align}
and $\Lambda$ is trace preserving if and only if
\begin{align}
    \operatorname{Tr}_{\mathrm{S'}}(J_\Lambda)=\mathbb{I}_{\mathrm{S}}.
\end{align}
where $S'$ denotes the auxiliary copy of the system $S$ in the definition of $J_{\Lambda}$ \cite{watrous2018theory}. These are semidefinite and affine constraints on $J_\Lambda$. Moreover, the constraint
\begin{align}
    \Lambda(\rho)=\sigma
\end{align}
is linear in $J_\Lambda$, while the conditions
\begin{align}
    \Lambda(\nu_i)\in F_s
\end{align}
can be written as
\begin{align}
    \Lambda(\nu_i)&=\sum_j p_j^i \nu_j,\\
    p_j^i&\geq 0,\\
    \sum_j p_j^i&=1,
\end{align}
which are again linear constraints. Therefore, the transformation problem is an SDP feasibility problem.
\endproof

\section{Proof of Theorem \ref{PQRT-Rep} }
Let $\mathcal R$ be a PQRT whose free states are generated by the vertex set $S_{\rm ver}=\{\rho_i\}_{i\in I}\subset D(\mathcal H)$, where $\mathcal H$ is a finite-dimensional Hilbert space. Let
$d=\dim(\mathcal H)$. Choose an orthonormal basis $B=\{\ket{k}\}_{k=1}^{d}$ of $\mathcal H$. Let $V=\mathrm{span}(B)$ be the corresponding abstract
complex vector space, and let $B'=\{f_k\}_{k=1}^{d}$
be the dual basis satisfying
\begin{equation*}
    f_k(\ket{\ell})=\delta_{k\ell}.
\end{equation*}
Thus $B'$ induces the standard scalar product for which $B$ is orthonormal. Therefore $\mathcal H$ is identified with the Hilbert-space realization $(V,B')$.

Now consider an arbitrary vertex state $\rho_i$. Since $\rho_i$ is a positive
semidefinite operator on a finite-dimensional Hilbert space, it admits a spectral decomposition $\rho_i=\sum_{k=1}^{d_i} \lambda_k^i
\ket{\psi_k^i}\bra{\psi_k^i}$,
where $\lambda_k^i\geq 0$, $\sum_k\lambda_k^i=1$, and
$d_i\leq d$. Define
\begin{equation*}
W_k^i:=\lambda_k^i\ket{\psi_k^i}\bra{\psi_k^i}.
\end{equation*}
Then each $W_k^i$ is positive semidefinite and has rank at most one. Hence $\rho_i=\sum_{k=1}^{d_i}W_k^i$.
Extending the sum by zero terms if necessary, we may write
\begin{equation*}
V_i:=\rho_i=\sum_{k=1}^{d}W_k^i,
\qquad \operatorname{rank}(W_k^i)\leq 1 .
\end{equation*}
Thus every vertex state $\rho_i$ is represented by an element
\begin{equation*}
    V_i=\sum_{k=1}^{d}W_k^i\in \mathcal{D}(V),
\end{equation*}

Therefore the whole vertex set $S_{\rm ver}$ is equivalently represented by
the collection
\[
\{V_i=\sum_{k=1}^{d}W_k^i\in \mathcal{D}(V):
\operatorname{rank}(W_k^i)\leq 1\}_{i\in I}.
\]
Together with the basis $B$ of $V$ and the dual basis $B'$ inducing the
Hilbert-space geometry, this gives the desired representation
\[
\left(
\{V_i=\sum_{k=1}^{d}W_k^i\in D(V):
\operatorname{rank}(W_k^i)\leq 1\}_{i\in I},
B,B'
\right).
\]

It remains to show uniqueness up to unitary transformations. Suppose that the
same PQRT is represented by another triple
\[
\left(
\{\widetilde V_i\}_{i\in I},\widetilde B,\widetilde B'
\right).
\]
The two choices $B$ and $\widetilde B$ are orthonormal bases of Hilbert-space
realizations of the same finite-dimensional vector space (with respect to $B'$ and $\widetilde{B'}$). Hence there exists a
unitary operator $U$ such that
\begin{equation*}
    U\ket{k}=\ket{\widetilde k}
\end{equation*}
for every basis element. Consequently, the corresponding density operators
satisfy
\begin{equation*}
    \widetilde V_i=U V_i U^\dagger
\end{equation*}
for all $i\in I$. Thus any two such representations differ only by a unitary change of basis.
Hence the representation is unique up to unitary transformation.
\endproof

\section{Proof of Theorem \ref{No-Iso}}

  Assume for the sake of contradiction that there exists a CPTP isomorphism $(M,Z)$ between $\mathcal{R}_{1}$ and $\mathcal{R}_{2}$. To begin, take an arbitrary state $\proj{v}$. We claim $M(\proj{v})$ is a pure state. Assume for the sake of contradiction that $M(\proj{v})$ is a mixed state. Therefore one can write:
      \begin{equation}
          M(\proj{v})=P\sigma + (1-P)\rho_{ver,2}^{i}
      \end{equation}
      where $0<P<1$ and $\rho_{ver,2}^{i}$ is a vertexal state in $\mathcal{R}_{2}$. Since $M$ is injective and linear, we have:
      \begin{equation}
          M^{-1}(P\sigma + (1-P)\rho_{ver,2}^{i})=P M^{-1}(\sigma)+(1-P)M^{-1}(\rho_{ver,2}^{i}).
      \end{equation}
      Note that $M^{-1}$ is linear on the range of $M$. By the definition of the isomorphism we must have $M^{-1}(\rho_{ver,2}^{i})=\rho_{ver,1}^{i}$ as the map $M$ must map the vertexal states to the vertexal states. Thus, we can write:
      \begin{equation}
          P M^{-1}(\sigma)+(1-P)\rho_{ver,1}^{i}=\proj{v}.
      \end{equation}
      But the right hand side is a rank-1 state which is a contradiction. Hence, $M(\proj{v})$ is pure.

     Next, take  an arbitrary state $\proj{v}$ with $\ket{v}=\frac{\ket{0}+\ket{1}}{\sqrt{2}}$, where $\{\ket{0},\ket{1}\}$ forms an orthonormal system and $\mathcal{H}_{2\text{-}D} := \operatorname{span}\{\ket{0},\ket{1}\}$ is a $2$-dimensional subspace of the vector space. We may consider the Pauli matrices for this subspace. We have
\begin{equation*}
\proj{v}
= \frac{1}{2}\proj{0}
+ \frac{1}{2}\proj{1}
+ \frac{1}{2}\bigl(\ket{0}\!\bra{1}+\ket{1}\!\bra{0}\bigr)
\end{equation*}
and we can represent the term $\ket{0}\!\bra{1}+\ket{1}\!\bra{0}$ by the
Pauli matrices acting on $\mathcal{H}_{2\text{-}D}$ as follows:
\begin{equation*}
\ket{0}\!\bra{1}+\ket{1}\!\bra{0}
= \sigma_x
= \proj{+} - \proj{-}.
\end{equation*}
Thus, we can write
\begin{equation*}
\proj{v}
= \frac{1}{2}\proj{0}
+ \frac{1}{2}\proj{1}
+ \frac{1}{2}\bigl(\ket{+}\!\bra{+}-\ket{-}\!\bra{-}\bigr).
\end{equation*}
Now let us apply $M$ on this state:
\begin{align*}
    M(\proj{v})
=
\frac{1}{2}&M(\proj{0})
+\frac{1}{2}M(\proj{1})
\\
&+\frac{1}{2}M(\ket{+}\!\bra{+})
-\frac{1}{2}M(\ket{-}\!\bra{-}).
\end{align*}
Since $M$ maps any pure state to a pure state, we must have
\begin{align*}
M(\proj{v})
&=
\frac{1}{2}\proj{0'}
+\frac{1}{2}\proj{1'}
\\&\quad +\frac{1}{2}\ket{+'}\!\bra{+'}
-\frac{1}{2}\ket{-'}\!\bra{-'}\\
&=
\proj{\tilde v}.
\end{align*}
and
\begin{align*}
\frac{1}{2}\proj{0'}
+\frac{1}{2}\proj{1'}
+\frac{1}{2}\proj{+'}
=
\proj{\tilde v}
+\frac{1}{2}\ket{-'}\!\bra{-'}.
\end{align*}
The right-hand side is at most rank $2$. Assume it is rank $1$.
Then the left-hand side must also be rank $1$, which means $M(\proj{0}) = M(\proj{+})$,
which is a contradiction because $M$ is injective.
Therefore, $\proj{\tilde v}+\frac{1}{2}\proj{-'}
$
is rank $2$ and consequently
\begin{align*}
\frac{1}{2}\proj{0'}
+\frac{1}{2}\proj{1'}
+\frac{1}{2}\proj{+'}
\end{align*}
must be rank $2$. For this, a necessary condition is that $\proj{+'}
\in
\operatorname{span}\{
\ket{0'}\!\bra{1'},
\ket{1'}\!\bra{0'},
\ket{0'}\!\bra{0'},
\ket{1'}\!\bra{1'}
\}$. This means that also the right-hand side must belong to the same span.
Therefore, we conclude that $
\ket{v'}=\frac{\ket{0'}+\ket{1'}}{\sqrt{2}}
$ and $
\braket{0'}{1'} = 0$. Hence, by arbitrariness of $\ket{v}$, $M$ maps any orthonormal basis to an orthonormal basis. Thus $\mathcal{R}_1$ and $\mathcal{R}_2$ are isomorphic.

\section{Proof of Theorem \ref{CP-Iso}} \label{CP-Iso-App}
    We will construct a CP- isomorphism between the two arbitrary polytopic resource theories:
    \begin{align}
    S_{ver,1}^{N}&= \{\proj{\psi_{i}}: i\in\{1,...,N\}\} \notag \\
     S_{ver,2}^{N}&= \{\proj{\phi_{i}}: i\in\{1,...,N\}\}.
\end{align}
 By the proposition \ref{GeoEquiv-Pure}, we can equivalently consider that the first resource theory is defined via $S_{ver,2}^{N}$ in the Hilbert space defined by the pair $(V,B'')$. Note that in order to use the proposition, we need to require that $|S_{ver}|\leq d$. Let $D_{1}=\text{conv}(\ket{v}\otimes f_{v}: f_{v}(\ket{v})=1)$, $D_{2}=\text{conv}(\ket{w}\otimes g_{w}:g_{w}(\ket{w})=1)$ and $T$ be the map defined by $T(f_{i})=g_{i}$. We claim that the map $M:D_{1}\to D_{2}$ defined by:
\begin{align}
    M(\rho)=\frac{[I\otimes T](\rho)}{\Tr([I\otimes T](\rho))},\text{ }\forall \rho \in D_1
\end{align}
    is the isomorphism map for the states. Indeed the map is surjective as for any extremal state $\ket{v}\otimes g_{v}$ in $D_{2}$, there is $\alpha \ket{v}\otimes f_{v} \in D_1$ for some $\alpha >0$ which is mapped to $\ket{v}\otimes g_{v}$ by $I\otimes T$. Furthermore, the inverse of $M$ exists and is:
    \begin{align}
        M^{-1}(\sigma)=\frac{[I\otimes T^{-1}](\sigma)}{\Tr([I\otimes T^{-1}](\sigma))},\text{ }\forall \sigma \in D_2.
    \end{align}
    Note that $T$ is linear and invertible. Hence $M$ is a bijection and an isomorphism. 
    
    We will show that the map $Z$ defined by $Z(\Lambda):= M'\Lambda M^{'-1}$ is the operation CP-isomorphism where $M'=I\otimes T$. By the definition, we have:
    \begin{align}
        Z(\lambda)(M(\rho))\propto M(\Lambda(\rho)).
    \end{align}
   Note that $\Tr(I\otimes T(\cdot))$ is positive on $D_{1}$. It remains to show that $M'\Lambda M'^{-1}$ is completely positive. To this purpose, we use the Choi-Jamio{\l}kowski operator of $Z(\lambda)$ \cite{watrous2018theory}:
    \begin{align}
        I\otimes M'\lambda M'^{-1} (\ket{\Psi}_2\bra{\Psi})
    \end{align}
    where $\ket{\Psi}_{l}\bra{\Psi}=\sum \ket{i}\!\bra{j}_{l}\otimes \ket{i}\!\bra{j}_{l}$ and $\{\ket{i}\}$ is an orthonormal basis with respect to $\{\bra{j}_{l}\}$. As $M'$ is defined on the whole space (not only $D_1$) and is invertible, we can write:
    \begin{align}
        [M^{'-1}M'&\otimes M'\Lambda M'^{-1}] (\ket{\Psi}\bra{\Psi})\notag\\&=[M^{'}\otimes M'\Lambda] (\sum \ket{i}\!\bra{j}_{1}\otimes \ket{i}\!\bra{j}_{1}) \notag \\
        &=
        [M^{'}\otimes M^{'}\Lambda]\ (\ket{\Psi}_{1}\bra{\Psi})\notag \\
        &=[(M'\otimes M')(I\otimes \Lambda)](\ket{\Psi}_{1}\bra{\Psi}).
    \end{align}
    Since $I\otimes \Lambda$ is a CPTP map, we may generally write $(I\otimes \Lambda)(\ket{\Psi}_{1}\bra{\Psi})=\sum P_{K}\ket{K}_{1}\bra{K}$. Hence;
    \begin{align}
        (M'\otimes M')(\sum P_{K}\ket{K}_{1}\bra{K})=\sum P_{K}\ket{K}_{2}\bra{K}>0.
    \end{align}
    and the proof is complete.

\section{Proof of Theorem \ref{Stoch-Iso}}

Assume that $\rho$ can be converted to $\sigma$ in the first resource theory. Then there exists a free operation $\Lambda\in F_{O,1}$ such that $
\Lambda(\rho)=\sigma$. Since $(M,Z)$ is a CP-homomorphism, the intertwining condition implies that for every $\omega\in D(\mathcal H_1)$, $
Z(\Lambda)(M(\omega))=\alpha\,M(\Lambda(\omega))$
for some $\alpha>0$. In particular, for $\omega=\rho$ we get $Z(\Lambda)(M(\rho))=\alpha\,M(\sigma)$. Now define
\begin{align}   
c_\Lambda:=\max_{\tau\in D(\mathcal H_2)} \Tr\bigl(Z(\Lambda)(\tau)\bigr).
\end{align}
Since $D(\mathcal H_2)$ is compact and $\tau\mapsto \Tr(Z(\Lambda)(\tau))$ is continuous, the maximum exists.Define
\begin{align}
\widetilde{\Lambda}:=\frac{1}{c_\Lambda}Z(\Lambda).
\end{align}
This map is linear and completely positive. Moreover, for every $\tau\in D(\mathcal H_2)$,
\begin{align}
\Tr(\widetilde{\Lambda}(\tau))
=
\frac{1}{c_\Lambda}\Tr(Z(\Lambda)(\tau))
\le 1.
\end{align}
Hence $\widetilde{\Lambda}$ is trace-nonincreasing, so it is a valid stochastic operation. Finally,
\begin{align}
\widetilde{\Lambda}(M(\rho))
=
\frac{1}{c_\Lambda}Z(\Lambda)(M(\rho))
=
\frac{\alpha}{c_\Lambda}M(\sigma).
\end{align}
Therefore $\widetilde{\Lambda}$ converts $M(\rho)$ to $M(\sigma)$ with nonzero success probability $\alpha/c_\Lambda$. Hence $M(\rho)$ is stochastically convertible to $M(\sigma)$ in the second resource theory.

\section{Proof of Proposition \ref{Non-Existence-Golden}}
We prove the claim by counterexample. Consider the multipartite entanglement resource  theory with fixed number of parties, with free operations to be fully separability-preserving (FSP) maps. If a golden state existed in this theory, then there would be a unique state that can be transformed
by free operations into every other resource state. But this is not the case: in the multipartite FSP resource theory there is, in general,
no unique maximally entangled state. Consequently, no single state plays the role of a universal resource from which all other resource states are obtainable by free operations \cite{Contreras-Tejada_Palazuelos_deVicente_2019}.
\endproof

\section{Proof of Proposition \ref{Golden-Prop}}

      (1)  We define $\mathcal{F}_{M}(\rho):=\max_{\sigma \in F_{s}}
        \mathcal{F}(\rho,\sigma)$ where $\mathcal{F}(\cdot,\cdot)$ is the fidelity function. As our resource theory is convex with a compact set of free states, $\mathcal{F}_{M}$ is a Fidelity based monotone and it achieves its minimum $\mathcal{F}_{M,\min}$ at a pure state.  Assume the golden state $\rho_{G}$ is not pure, therefore $\mathcal{F}_{M}(\rho_{G})-\mathcal{F}_{M,\min}=\epsilon > 0$, this means that by implementing the free operations on $\rho_{G}$ one cannot reach to the states with $\mathcal{F}_{M}<\mathcal{F}_{M}(\rho_{G})$ which is a contradiction as we assumed that $\rho_{G}$ is a golden state.

        (2) According to the first part of the lemma, a golden state must satisfy:
        \begin{align}
            \mathcal{F}_{M}(\ket{\psi_{G}})&=\min_{\ket{\psi}}\max_{\{J_{j}\}}\sum_{i}J_{i}\braket{\psi|\rho_{i}}{\psi}\notag\\&=\min_{\ket{\psi}}\max_{\{J_{j}\}}\sum_{i}J_{i} \text{Tr}(\rho_{i}\ket{\psi}\!\bra{\psi})
        \end{align}
        where $\{J_{k}\}$ is a probability distribution over $\{\rho_{k}\in S_{ver}\}$. As $\text{Tr}(\cdot)$ is a linear function, we may equivalently write:
        \begin{align}
             \mathcal{F}_{M}(\ket{\Psi_{G}})=\min_{\rho}\max_{\{J_{j}\}}\sum_{i}J_{i} \text{Tr}(\rho_{i}\rho)
        \end{align}
        Let denote $\text{Tr}(\rho_{i}\rho):=C_{i}^{\rho}>0$. Note that for the purpose of this proposition, the states $\rho_{i}$ are pure. Thus, we have:
         \begin{align}
            \mathcal{F}_{M}(\ket{\Psi_{G}})=\min_{\rho}C_{\max}^{\rho}
        \end{align}
        and the golden state satisfies $
    \ket{\Psi_G} = \arg\min_{\psi} C_{\max}^\psi.$
        
        We now show that at the optimum all overlaps are equal. Suppose, for contradiction, that the values $\{C_i^{\Psi_G}\}$ are not all equal. Then there exists $k$ such that  $C_k^{\Psi_G} = C_{\max}^{\Psi_G} > C_j^{\Psi_G}$ for some $j$. Using continuity of the map $\ket{\psi} \mapsto (C_1^\psi,\dots,C_n^\psi)$ and the connectedness of the pure state space, there exists a small perturbation $\ket{\psi'}$ such that $C_k^{\psi'}$ decreases while the other components increase only infinitesimally. For sufficiently small perturbations, this implies $
    C_{\max}^{\psi'} < C_{\max}^{\Psi_G},
$ contradicting optimality of $\ket{\Psi_G}$.

\section{Proof of Proposition \ref{Instrument-Monotone} }
Again from Eq. \ref{albela}, there exist $\sigma\in F_s$ and $\delta\in D(\mathbb{C}^d)$ for a state $\rho\in D(\mathbb{C}^d)$ such that 
\begin{equation}\label{gale}
    \rho=(1+\mathcal{M}_R(\rho))\sigma-\mathcal{M}_R(\rho)\delta.
\end{equation}

This implies,
\begin{eqnarray}
\frac{\Gamma_n(\rho)}{\Tr(\Gamma_n(\rho))}&=&(1+\mathcal{M}_R(\rho))\frac{\Gamma_n(\sigma)}{\Tr(\Gamma_n(\sigma))}\frac{\Tr(\Gamma_n(\sigma))}{\Tr(\Gamma_n(\rho))}\notag\\&&-\mathcal{M}_R(\rho)\frac{\Gamma_n(\delta)}{\Tr(\Gamma_n(\delta))}\frac{\Tr(\Gamma_n(\delta))}{\Tr(\Gamma_n(\rho))}.
\end{eqnarray}

Thus,
\begin{eqnarray}
\frac{\Gamma_n(\sigma)}{\Tr(\Gamma_n(\sigma))}&=&\frac{\frac{\Gamma_n(\rho)}{\Tr(\Gamma_n(\rho))}+\mathcal{M}_R(\rho)\frac{\Tr(\Gamma_n(\delta))}{\Tr(\Gamma_n(\rho))}\frac{\Gamma_n(\delta)}{\Tr(\Gamma_n(\delta))}}{(1+\mathcal{M}_R(\rho))\frac{\Tr(\Gamma_n(\sigma))}{\Tr(\Gamma_n(\rho))}} \nonumber \\
&=&\frac{\frac{\Gamma_n(\rho)}{\Tr(\Gamma_n(\rho))}+\mathcal{M}_R(\rho)\frac{\Tr(\Gamma_n(\delta))}{\Tr(\Gamma_n(\rho))}\frac{\Gamma_n(\delta)}{\Tr(\Gamma_n(\delta))}}{(1+\mathcal{M}_R(\rho)\frac{\Tr(\Gamma_n(\delta))}{\Tr(\Gamma_n(\rho))})}.
\end{eqnarray}
The second equality follows from the equation obtained by applying $\Gamma_n$ on both sides of Eq. \eqref{gale}. 
Since $\frac{\Gamma_n(\sigma)}{\Tr(\Gamma_n(\sigma))}$ belongs to the free states, therefore, by definition Eq.~\eqref{albela},
\begin{eqnarray}
  && \mathcal{M}_R\left( \frac{\Gamma_n(\rho)}{\Tr(\Gamma_n(\rho))}\right)\leq \mathcal{M}_R(\rho)\frac{\Tr(\Gamma_n(\delta))}{\Tr(\Gamma_n(\rho))} \nonumber \\
   \implies && \sum_n \Tr(\Gamma_n(\rho))\mathcal{M}_R\left( \frac{\Gamma_n(\rho)}{\Tr(\Gamma_n(\rho))}\right)\leq \mathcal{M}_R(\rho).
\end{eqnarray}
\endproof

\medskip

\section{Proof of the Proposition \ref{Negativity-Monotone}}
For every free operation $\lambda\in F_O$ and every extremal free state
$\proj{\nu_i}$, we have
\begin{equation}
    \lambda(\proj{\nu_i})\in F_s=\operatorname{conv}\{\proj{\nu_j}\}_j .
\end{equation}
Hence there exist probabilities $p_{j|i}\geq 0$ with
$\sum_j p_{j|i}=1$ such that
$ \lambda(\proj{\nu_i})=\sum_j p_{j|i}\proj{\nu_j}$. Therefore, if $\rho=\sum_i\alpha_i(\rho)\proj{\nu_i}$, then
\begin{equation}
    \lambda(\rho)=\sum_i\alpha_i(\rho)\lambda(\proj{\nu_i})=\sum_{i,j}\alpha_i(\rho)p_{j|i}\proj{\nu_j}.
\end{equation}
By uniqueness of the expansion in the basis $\{\proj{\nu_j}\}_j$, we get $\alpha_j(\lambda(\rho))=\sum_i p_{j|i}\alpha_i(\rho)$. Thus the coefficient vector transforms as $ \alpha(\lambda(\rho))=P_\lambda \alpha(\rho)$, where $P_\lambda=(p_{j|i})_{j,i}$ is a column-stochastic matrix. Now write
\begin{equation}
     \alpha_i(\rho)=\alpha_i^+(\rho)+\alpha_i^-(\rho),
    \end{equation}
where $\alpha_i^+(\rho)\geq 0$ and $\alpha_i^-(\rho)\leq 0$. Since
$P_\lambda$ has non-negative entries, we have
\begin{equation}
    P_\lambda \alpha(\rho)
    = P_\lambda \alpha^+(\rho)+P_\lambda \alpha^-(\rho),
\end{equation}
with
\begin{equation}
    P_\lambda \alpha^+(\rho)\geq 0,
    \qquad
    P_\lambda \alpha^-(\rho)\leq 0 .
\end{equation}
Since $|P_{\lambda}\alpha^{-}(\rho)|\leq \sum_{k}|\alpha_{k}^{-}(\rho)|$ and a nonnegative vector can only decrease the absolute value of the negative parts, we must have
\begin{equation}
    \mathcal M_N(\lambda(\rho)) = -\sum_j \alpha_j^-(\lambda(\rho))
    \leq
    -\sum_j (P_\lambda \alpha^-(\rho))_j .
\end{equation}
Therefore
\begin{equation}
    \mathcal M_N(\lambda(\rho))\leq \mathcal M_N(\rho).
\end{equation}
\endproof

\end{document}